\documentclass[12pt]{article}
{\begin{Sbox}\begin{minipage}}%
{\end{minipage}\end{Sbox}\fbox{\TheSbox}}%

\usepackage[psamsfonts]{amssymb}
\usepackage{amsmath, amsthm, anysize, enumerate, color, url}
\usepackage{graphicx}
\usepackage{subfigure}
\usepackage{mathrsfs}
\usepackage{float}
\newtheorem{theorem}{Theorem} 

\newtheorem{lemma}{Lemma} 
\newtheorem{proposition}{Proposition} 
\newtheorem{proof of lemma}{Proof of Lemma}
\newtheorem{proof of theorem}{Proof of Theorem}
\usepackage[colorlinks,
                   linkcolor=red,
                  anchorcolor=blue,
                  citecolor=blue
                  ]{hyperref}

\usepackage{mathrsfs}

\theoremstyle{definition}
\newtheorem{remark}{Remark}  

\newcommand{\R}{\mathbb{R}}

\usepackage{setspace}

\numberwithin{equation}{section}
\theoremstyle{plain}

\usepackage{indentfirst}
\usepackage{natbib}
\begin{document}

\title{Affiliation networks with an increasing degree sequence}

\author{
Yong Zhang\thanks{Department of Statistics$^{*,\dag,\ddag,\S}$ and Hubei Key Laboratory of Mathematical Sciences$^{\ddag,\S}$,  Central China Normal University, Wuhan, 430079, China.
\texttt{Emails:} $^*$zhang$\_$yong@mails.ccnu.edu.cn,
$^\dag$qianxiaodi@mails.ccnu.edu.cn,
$^\ddag$qinhong@mail.ccnu.edu.cn,
$^\S$tingyanty@mail.ccnu.edu.cn.} 
\hspace{15mm}
Xiaodi Qian$^\dag$
\hspace{15mm}
Hong Qin$^\ddag$
 \hspace{15mm}
Ting Yan$^\S$
\\~~\\
Central China Normal University
}
\date{\empty}
\maketitle
\begin{abstract}

Affiliation network is one kind of two-mode social network with two different
sets of nodes (namely, a set of actors and a set of social events) and edges
representing the affiliation of the actors with the social events.
Although a number of statistical models are proposed to analyze affiliation networks,
the asymptotic behaviors of the estimator are still unknown or have not been properly explored.
In this paper, we study an affiliation model with the degree sequence as the exclusively natural sufficient statistic
in the exponential family distributions.
We establish the uniform consistency and asymptotic normality of the maximum likelihood estimator when the
numbers of actors and events both go to infinity.
Simulation studies and a real data example demonstrate our theoretical results.

\vskip 5 pt \noindent
\textbf{Key words}: Affiliation networks; Asymptotic normality; Consistency; Maximum likelihood estimators \\

{\noindent \bf Mathematics Subject Classification:} 	62E20, 62F12.
\end{abstract}

\vskip 5pt

Running title:  Affiliation networks

\section{Introduction}
Affiliation network is one kind of two-mode social network that consists of
two different types of sets of nodes, namely, a set of actors and a set of social events.
The network edges indicate the affiliation of actors with social events.
Such network data are commonly used to represent memberships between social organizations and their members,
for example, the affiliation of the researchers to the academic institutions or 
interlocking directors to companies or actors to movies. Other scenarios include ceremonial events attended by members
faculty sit, social events people attend, trade partners of major
oil exporting nations and so on.

In affiliation networks, the actors are brought together to jointly participate in social events.
Joint participation in events not only provides the opportunity for actors to interact,
but also increases the probability that links (e.g., friendship) between actors form.
For example, belonging to the same organizations (boards of directors,
political party, labor union, and so on) provides the opportunity for people to
meet and interact, and thus a link between individuals is more easily to form in these circumstances.
Similarly, when actors participate in more than one event, two events are connected through these actors.
There has been increasing interest in analyzing affiliated network data in recent years.
A number of approaches are proposed to analyze or model affiliation network data [e.g., \cite{conyon2004small}, \cite{robins2004small}, \cite{snijders2013model}].
\cite{iacobucci1990social} proposed a $p_2$ exponential family distribution using the degree sequence as the sufficient statistics
to model the weighted affiliation network, which is a close relative of $p_1$ model introduced  by \cite{Holland:Leinhardt:1981}.
\cite{latapy2008basic} extended the basic network statistics used to analyze one-mode networks to give a description
of analysis for two-mode networks systematically.
\cite{snijders2013model} proposed a stochastic actor-oriented model for the co-evolution of two-mode and one mode networks.
By extending exponential random graph models for the one-mode networks,
\cite{wang2009exponential} proposes a number of two-mode specifications as the sufficient statistics in exponential family graph models for
two-mode affiliation networks and compared the goodness of fit results obtained using the maximum likelihood and pseudo-likelihood
approaches by simulation. 

At present, little theoretical results are obtained in affiliation network models although
many properties of statistical models for one-mode networks are derived [e.g., \cite{Chatterjee:Diaconis:2013}, \cite{Shalizi:Rinaldo:2013}, \cite{Bhattacharyya2016}].
Even in the aforementioned simple \cite{iacobucci1990social}'s model by assuming
that all edges are independent,
the asymptotic properties of the maximum likelihood estimator (MLE) are still not addressed due to a growing dimension of parameter space.
In this paper, we study the asymptotic properties of the MLE in
an affiliation model with the degree sequence as the exclusively natural sufficient statistic
in the exponential family distributions.
This model is identical to the \cite{iacobucci1990social} model for unweighed edges (i.e., binary edges).
We establish the uniform consistency and asymptotic normality of the maximum likelihood estimator when the
number of actors and events both go to infinity. A key step to the proof is that we make use of
the approximate inverse of the Fisher information matrix with small approximation errors,
which is the extension of that used in \cite{yan2016asymptotics}.

The rest of this paper is organized as follows. We introduce the model in Section \ref{section 2}.
In Section \ref{section 3}, we present the asymptotic results including the uniform consistency
and asymptotic normality of the MLE. Simulation studies and a real application to the student extracurricular affiliation data
are given in Section \ref{section 4}.
Some further discussions are devoted to Section \ref{section 5}.
All the proofs are putted in the appendix.

\section{Model}
\label{section 2}

\subsection{Notations}

Let $A$ be an event set with $m$ events denoted by $\{1, \ldots, m\}$, and $P$ be an actor set with $n$ actors denoted by $\{1, \ldots, n\}$.
An affiliation network ${\cal{G}}(m,n)$ records the affiliation of each actor with each event in an affiliation matrix $X=(x_{i,j})_{m\times n}$;
$x_{i,j}=1$ if actor $i$ is affiliated with event $j$ and $x_{i,j}=0$ otherwise.
Each row of $X$ describes an actor's affiliation with the events and each column describes the memberships of the event.
In practice, $n$ is usually large and $m$ relatively small. Therefore, we assume $m\le n$ hereafter.
The affiliation network $\mathcal{G}(m,n)$ can also be represented in a bipartite graph, in which
the nodes are partitioned into two subsets for the actors and the events and the edges exist only between pairs of nodes belonging to different subsets.
In bipartite graphs, no two actors are adjacent and no two events are adjacent.
If pairs of actors are reachable, it is only via paths containing one or more events.
Similarly, if pairs of events are reachable, it is only via paths containing one or more actors.

Define $d_i=\sum_{j=1}^{n}x_{i,j}$ as the degree of vertex $i\in A$ and $\mathbf{d}=(d_1,\dots,d_m)^\top$.
Similarly, define $b_j=\sum_{i=1}^{m}x_{i,j}$ as the degree of vertex $j\in P$ and $\mathbf{b}=(b_1,\dots,b_n)^\top$.
The pair $\{\mathbf{d,b}\}$ is the degree sequence of the affiliation network ${\cal{G}}(m,n)$.

\subsection{An affiliation network model}
In this subsection, we present an exponential random bipartite graph model for
affiliation networks with the degree sequence as the exclusively natural sufficient statistic.
The probability mass function on the affiliation network $\mathcal{G}(m,n)$ is assumed to be of exponential form:
\begin{equation}\label{2.1}
P({\mathcal{G}(m,n)})=\exp({\boldsymbol{\alpha}}^\top \mathbf{d}+\boldsymbol{\beta}^\top \mathbf{b}-\mathrm{Z}(\boldsymbol{\alpha},\boldsymbol{\beta})),
\end{equation}
where $\mathrm{Z}(\boldsymbol{\alpha},\boldsymbol{\beta})$ is the normalizing constant,
$\boldsymbol{\alpha}=(\alpha_1,\dots,\alpha_m)^\top$ and $\boldsymbol{\beta}=(\beta_1,\dots,\beta_n)^\top$ are parameter vectors.
Each affiliation network with the same degree sequence is equally judged.
The parameter $\alpha_i$ quantifies the popularity of the event $i$ and $\beta_j$ quantifies
the activity of the actor $j$ to participate in events.
Note that
\[\exp({\boldsymbol{\alpha}}^{\top}\mathbf{d}+\boldsymbol{\beta}^{\top}\mathbf{b})=\exp\left(\sum^{m}_{i=1}\sum^{n}_{j=1}(\alpha_i+\beta_j)x_{i,j}\right)
=\prod^{m}_{i=1}\prod^{n}_{j=1}\exp((\alpha_i+\beta_j)x_{i,j}),
\]
which implies that the $mn$ random variables $x_{i,j}$ are mutually independent and $\mathrm{Z}(\boldsymbol{\alpha},\boldsymbol{\beta})$ can be expressed as
\[
\mathrm{Z}(\boldsymbol{\alpha},\boldsymbol{\beta})=\sum_{i=1}^m\sum_{j=1}^n
\log( 1 + \exp(\alpha_i+\beta_j) ).
\]
Therefore, $x_{i,j}$, $1\leq i\leq m,1\leq j\leq n$ are mutually independent Bernoulli random variables with the success probability:
\begin{equation}\label{eq:probability}
{\mathbb{P}}(x_{ij}=1)=\frac{e^{\alpha_{i}+\beta_{j}}} {1+e^{\alpha_{i}+\beta_{j}}},
\end{equation}
which is the \cite{iacobucci1990social} model for affiliated networks with binary edges.

Since the sample is just one realization of the bipartite random graph, the density or probability mass function \eqref{2.1} is also the likelihood function.
If one transforms $(\boldsymbol{\alpha},\boldsymbol{\beta})$ to $(\boldsymbol{\alpha}-c,\boldsymbol{\beta}+c)$, the likelihood does not change. Following \cite{yan2016asymptotics}, we set $\beta_n=0$ for the identifiability of the parameter.

The consistency and asymptotic normality of the MLE in the exponential random graph model with
the degree sequence for undirected one-model networks or the bi-degree sequence for directed one-mode networks
have been established recently [\cite{chatterjee2011}, \cite{hillar2013maximum}, \cite{yan2013central}, \cite{yan2016asymptotics},
\cite{Yan:Qin:Wang:2016}].
For the bipartite networks, the asymptotic theory for the MLE in the model \eqref{2.1} has not been explored.
The model is closely related to the Rasch model [\cite{Rasch:1960}] for dichotomous item response experiments, which assumes that
item $i$ correctly gives a response to subject $j$ with probability $\exp( \alpha_i - \beta_j )/(1+\exp( \alpha_i - \beta_j))$.
By assuming that all the parameters in the Rasch model are bounded, \cite{Haberman:1977}
proved the consistency and asymptotic normality of the MLE when the number of items and subjects goes to infinity simultaneously.

\section{Asymptotic results} \label{section 3}

Let ${\boldsymbol{\theta}}=(\alpha_1, \ldots, \alpha_m, \beta_1, \ldots, \beta_{n-1})^\top$ and $\mathbf{g}=(d_1, \ldots, d_m, b_1, \ldots, b_{n-1})^\top$.
The log-likelihood function is
\[
\ell({\boldsymbol{\theta}})=\sum^{m}_{i=1}\alpha_{i}d_{i}+\sum^{n-1}_{j=1}\beta_{j}b_{j}-\sum_{i=1}^{m}\sum_{j=1}^{n}\log(1+ e^{\alpha_{i} + \beta_{j}}).
\]
The likelihood equations are:
\begin{equation}\label{eq:2.2}\setlength\arraycolsep{2pt}
\begin{array}{rcl}
d_i&=&\sum_{j=1}^{n}\frac{e^{\alpha_{i} + \beta_{j}}} {1+e^{\alpha_{i} + \beta_{j}}},~~~i=1, \ldots, m,\\
b_j&=&\sum_{i=1}^{m}\frac{e^{\alpha_{i} + \beta_{j}}} {1+e^{\alpha_{i} + \beta_{j}}},~~~j=1, \ldots, n-1.
\end{array}
\end{equation}
Let $\widehat{{\boldsymbol{\theta}}}=(\widehat{\alpha}_{1},\dots,\widehat{\alpha}_{m},\widehat{\beta}_{1},\cdots,\widehat{\beta}_{n-1})^\top$ be the MLE of ${\boldsymbol{\theta}}$ and $\widehat{\beta}_{n}=0$. If $\widehat{{\boldsymbol{\theta}}}$ exists, then it is the solution to
the system of equation \eqref{eq:2.2}.

Let $V=(v_{i,j})$ be the Fisher information matrix of the parameter vector $\boldsymbol{\theta}$, which is a diagonal dominant matrix with nonnegative entries.
The diagonal elements of $V$ are
\[
v_{i,i}=\sum^{n}_{j=1}\frac{e^{\alpha_{i}+\beta_{j}}}{(1+e^{\alpha_{i}+\beta_{j}})^2},~i = 1,\dots, m,~~
v_{m+j, m+j}=\sum^{m}_{i=1}\frac{e^{\alpha_{i}+\beta_{j}}}{(1+e^{\alpha_{i}+\beta_{j}})^2},~j = 1,\dots, n.
\]
Motivated by the approximate inverse proposed by \cite{yan2016asymptotics} for the Fisher information matrix in the
directed one-mode network model involved with the bi-degree sequence, we proposed a generalized simple matrix $S=(s_{i,j})$ to approximate the $V^{-1}$, which is defined as
\begin{equation}\label{eq:Smatrix}
s_{i,j}=
\left\{
\begin{array}{ll}
\frac{\delta_{i,j}}{v_{i,i}}+\frac{1}{v_{m+n,m+n}}, & {i,j=1,\dots,m,} \\
-\frac{1}{v_{m+n,m+n}}, & {i=1,\dots,m,~j=m+1,\dots,m+n-1,}  \\
-\frac{1}{v_{m+n,m+n}}, & {i=m+1,\dots,m+n-1,j=1,\dots,m,}\\
\frac{\delta_{i,j}}{v_{i,i}}+\frac{1}{v_{m+n,m+n}}, & {i,j=m+1,\dots,m+n-1,}
  \end{array}
\right.
\end{equation}
where $\delta_{i,j}=1$ when $i=j$ and $0$ otherwise. 
For a vector $\mathbf{x}=(x_1,\dots,x_n)^{\top}\in \R^n$,
denote the $\ell_{\infty}$ norm of $\mathbf{x}$ by $\|\mathbf{x}\|_{\infty}=\max_{1\leq i\leq n}| x_i |$.
We present the consistency of $\widehat{{\boldsymbol{\theta}}}$ here, whose proof is given in Appendix A.

\begin{theorem}\label{Theorem 1}
Assume that ${{\boldsymbol{\theta}}}^*\in \mathbb{R}^{m+n-1}$ with $\| {{\boldsymbol{\theta}}}^*\|_{\infty}\leq \tau\log n$, where $\tau\in (0, 1/24)$ is a constant, and that $X \sim \mathbb{P}_{{{\boldsymbol{\theta}}}^*}$, where $\mathbb{P}_{{{\boldsymbol{\theta}}}^*}$ denotes the probability distribution (\ref{2.1}) on $X$ under the parameter ${{\boldsymbol{\theta}}}^*$.
If $m/n=O(1)$, then as $n$ goes to infinity, with probability approaching one, the MLE $\widehat{{\boldsymbol{\theta}}}$
exists and satisfies
\[\|{\widehat{{\boldsymbol{\theta}}}}- {{\boldsymbol{\theta}}}^*\|_{\infty}=O_p\left(\frac{(\log n )^{1/2}e^{6\parallel {{\boldsymbol{\theta}}}^*\parallel_{\infty}}}{n^{1/2}}\right)=o_p(1). \]
Further, if the MLE exists, it is unique.
\end{theorem}

Next, we present the central limit theorem of $\widehat{{\boldsymbol{\theta}}}$, whose proof is given in Appendix B.

\begin{theorem}\label{Theorem:central}
Assume that $X \sim \mathbb{P}_{{{\boldsymbol{\theta}}}^*}$. If  $m/n=O(1)$ and $\| {{\boldsymbol{\theta}}}^*\|_{\infty}\leq \tau\log n$, where $\tau \in (0,1/36)$ then for any fixed $k\geq 1$, as $n\rightarrow \infty$, the vector consisting of the first $k$ elements of $(\widehat{{\boldsymbol{\theta}}}-{{\boldsymbol{\theta}}}^*)$ is asymptotically multivariate normal with mean zero and covariance matrix given by the upper left $k\times k$ block of $S^{*}$.
where $S^{*}$ is the matrix by replacing ${{\boldsymbol{\theta}}}$ in $S$ given in \eqref{eq:Smatrix} with its true value ${{\boldsymbol{\theta}}}^*$.
\end{theorem}

\begin{remark}\label{remark 1}
By Theorem \ref{Theorem:central}, for any fixed $i$, as $n\rightarrow \infty$ , the convergence rate of ${\widehat{{{\theta}}}}_i$ is $1/v_{i,i}^{1/2}$. Since $me^{-2\parallel{{\boldsymbol{\theta}}^{*}}\parallel_{\infty}}\leq v_{i,i}\leq n/4$, the rate of convergence is between $O(m^{-1/2}e^{\| {{\boldsymbol{\theta}}}^*\|_{\infty}})$ and $O(n^{-1/2})$.
\end{remark}

\section{Simulation studies}\label{section 4}
We carry out the numerical simulations to evaluate Theorem \ref{Theorem:central}.
Following \cite{yan2016asymptotics}, the parameter values take a linear form.
Specifically, we set $\alpha_{i+1}^{*}=(m-1-i)L/(m-1)$ for $i=0,1,\dots,m-1$,
and considered four different values $0$, $\log(\log(m))$, $\log(m)^{1/2}$ and $\log(m)$ for $L$.
For the parameter vector $\boldsymbol{\beta}$, let $\beta_{j+1}^{*}=(n-1-i)L/(n-1), j=0,1,\dots,n-2$ for simplicity and $\beta_{n}^{*}=0$ by default.

By Theorem \ref{Theorem:central}, $\hat{\xi}_{i,j}=[\hat{\alpha}_{i}-\hat{\alpha}_{j}-(\alpha^{*}_i-\alpha^{*}_j)]/(1/\hat{v}_{i,i}+1/\hat{v}_{j,j})^{1/2}$,
$\hat{\eta}_{i,j}=[\hat{\beta}_{i}-\hat{\beta}_{j}-(\beta^{*}_i-\beta^{*}_j)]/(1/\hat{v}_{m+i,m+i}+1/\hat{v}_{m+j,m+j})^{1/2}$, are asymptotically distributed as standard normal distribution, where $\hat{v}_{i,i}$ is the estimate
of ${v}_{i,i}$ by replacing $\boldsymbol{\theta}$ with $\boldsymbol{\widehat{\theta}}$.
We assess the asymptotic normality of $\hat{\xi}_{i,j},~\hat{\eta}_{i,j}$ using the quantile-quantile (QQ) plot.
Further, we record the coverage probability of the
$95\%$ confidence interval, the length of the confidence interval, and the frequency that the MLE does not exist.
Each simulation is repeated $10,000$ times.

We only simulate a single combination for $(m,n)$ with $m=100,n=200$, and present the QQ-plots of $\hat{v}_{ii}^{-1/2}(\hat{\alpha}_i-\alpha_i)$ and $\hat{v}_{ii}^{-1/2}(\hat{\beta}_i-\beta_i)$ in Figure \ref{figure-discrete-qq,alpha} and  Figure \ref{figure-discrete-qq,beta}, respectively.
The horizontal and vertical axes are the theoretical and empirical quantiles, respectively, and the
straight lines correspond to the reference line $y=x$.
In Figure \ref{figure-discrete-qq,alpha}, we can see that the empirical quantiles coincide with the theoretical ones very well. In Figure \ref{figure-discrete-qq,beta}, there are slight deviations when $L=( \log m)^{1/2}$.
When $L=(\log n)^{1/2}$, there are a little derivations on both tails of plots.
When $L=\log n$, the MLE does not exist in all repetitions, the QQ plots are not available in this case.

The coverage probability of the $95\%$ confidence interval for $\alpha_i-\alpha_j$ and $\beta_i-\beta_j$, the length of the confidence interval, and the frequency that the MLE did not exist, which are reported in Table \ref{Table 1:Estimatation of alpha} and Table \ref{Table 2:Estimatation of beta}, respectively. There are similar results for the two tables.
We can see that the length of estimated confidence interval increases as $L$ increases for fixed $n$, and decreases as $n$ increases for fixed $L$.
The coverage frequencies are all close to the nominal level $95\%$.
When $L=(\log n)^{1/2}$ (conditions in Theorem \ref{Theorem:central} no longer hold),
the MLE  does not exist with a positive probability; when $L=\log(n)$,
the MLE did not exist with $100\%$ frequencies.

{\it A data example}. We analyze a student extracurricular affiliation network data collected by Dan McFarland in 1996, which can be downloaded from \url{http://dl.dropbox.com/u/25710348/snaimages/mag_act96.txt}. It consists of $1295$ students (anonymized) and  $91$  student organizations in which they are members (e.g. National Honor Society, wrestling team, cheerleading squad, etc.).
In order to guarantee the existence of the MLE, we remove those $438$ individuals that don't belong to any organizations.
The MLEs of the parameters for remaining students and organizations and their standard errors as well as the $95\%$ confidence intervals  are reported in Tables \ref{Table 3:estimated influence parameters} and  \ref{Table 4:estimated influence parameters}.
The value of estimated parameters reflect the size of degrees.
For example, the largest five degrees in student organizations are $199, 157, 124, 93, 89$ for organizations ``Spanish.Club, Pep.Club, NHS, Latin.Club, Orchestra.Symphonic",  which also have the top five influence parameters $-0.32, -0.64, -0.94, -1.92, -1.34$.
On the other hand, the organizations with the five smallest influence parameters
$-4.60, -4.60, -4.60, -4.89, -4.89$  have degrees $4, 4, 4, 3, 3$.

\section{Summary and discussion} \label{section 5}
Statistical models for affiliation networks provide insight into the formulation of complex social
affiliation between actors and events. They also indirectly reflect how events create ties among actors and
the actors create ties among events. Meanwhile, the asymptotic inference in these models are challenge
like other network models due to that the structure of the network data is non-standard.
In this paper, we derive the uniform consistency and asymptotic normality of the MLE
in the exponential random bipartite graph models for affiliation networks with the degree sequence as the exclusively sufficient statistic.
The conditions imposed on $\tau$ that guarantee the good asymptotic properties of the MLE may not be best possible.
In particular, the conditions guaranteeing the asymptotic normality are stronger than those guaranteeing the consistency.
Simulation studies suggest that the conditions on $\tau$ might be relaxed.
The asymptotic behavior of the MLE depends not only on $\tau$, but also on the configuration of all the parameters.
We will investigate this problem in the future work.

We only consider dyadic independence assumption.
Like the model specifications in the exponential random graph models for the one-mode network, one
can add the counts of $k$-stars of various sizes as the model terms in the model \eqref{2.1} to model
the dependence structure. However, such model incurs the degeneracy problem [e.g., \cite{Chatterjee:Diaconis:2013}]
in which the generated graphs are almost full or empty. To overcome this disadvantage,
\cite{wang2009exponential} proposes a number of new two-mode specifications such as the k-two-paths and three-path statistics.
Although \cite{wang2009exponential} show that these model specifications have good performance by simulations,
the theoretical properties of the model are still unknown.

\section*{Appendix}
In this appendix, we will present the proofs for Theorems \ref{Theorem 1} and \ref{Theorem:central}.
We start with some preliminaries.
For an $n\times n$ matrix $J=(J_{i,j})$, $\|J\|_{\infty}$ denotes the matrix norm induced by the $\|\cdot\|_{\infty}$-norm on vectors in $\mathbb{R}^n$:
\[
\|J\|_{\infty}=\max_{\mathbf{x}\neq 0}\frac{\|J\mathbf{x}\|_{\infty}}{\|\mathbf{x}\|_{\infty}}=\max_{1\leq i\leq n}\sum_{j=1}^{n}| J_{i,j}|.
\]
Let $D$ be an open convex subset of $\mathbb{R}^n$. We say an $n\times n$  function matrix $F(\mathbf{x})$ whose elements $F_{ij}(\mathbf{x})$ are functions on vectors $\mathbf{x}$, is Lipschitz continuous on $D$ if there exists a real number $\lambda$ such that for any $\mathbf{v}\in R^n$ and any $\mathbf{x,y}\in D,$
\[
\| F(\mathbf{x})(\mathbf{v})-F(\mathbf{y})(\mathbf{v})\|_{\infty}
\leq \lambda\| \mathbf{x-y} \|_{\infty} \| \mathbf{v} \|_{\infty},
\]
where $\lambda$ may depend on $n$ but independent of $\mathbf{x}$ and $\mathbf{y}$.
For fixed $n$, $\lambda$ is a constant.

We introduce a class of matrices.
Given two positive numbers $q, Q$, we say the $(m+n-1)\times(m+n-1)$ matrix $V=(v_{i,j})$ belongs to the class $ {\cal L}_{m,n}(q,Q)$ if the following holds:
\begin{equation}\label{eq:Lmatrix}
\begin{array}{l}%
q\leq v_{i,i}-\sum_{j=m+1}^{m+n-1}v_{i,j}\leq Q,i=1,\dots,m;v_{m,m}=\sum_{j=m+1}^{m+n-1}v_{m,j},\\
v_{i,j}=0,i,j=1,\dots,m,i\neq j,\\
v_{i,j}=0,i,j=m+1,\dots,m+n-1,i\neq j,\\
q\leq v_{i,j}=v_{j,i}\leq Q,i=1,\dots,m,j=m+1,\dots,m+n-1,\\
v_{i,i}=\sum_{k=1}^{m}v_{k,i}=\sum_{k=1}^{m}v_{i,k},i=m+1,\dots,m+n-1.
\end{array}
\end{equation}
If $V\in{\cal L}_{m,n}(q,Q)$, then $V$ is a $(m+n-1)\times(m+n-1)$ diagonally dominant, symmetric
nonnegative matrix. Define $v_{m+n,i}=v_{i,m+n}:=v_{i,i}-\sum_{j=1}^{m+n-1}v_{i,j}$ for $i=1,\dots,m+n-1$ and $v_{m+n,m+n}=\sum_{i=1}^{m+n-1}v_{m+n,i}.$ Then $q\leq v_{m+n,i}\leq Q$ for $i=1,\dots,m,v_{m+n,i}=0$ for $i=m,m+1,\dots,m+n-1$ and $v_{m+n,m+n}=\sum_{i=1}^{m}v_{i,m+n}=\sum_{i=1}^{m}v_{m+n,i}.$
The Fisher information matrix of the parameter vector $\boldsymbol{\theta}$, $V$,
belongs to the matrix class $V\in{\cal L}_{m,n}(q,Q)$.
The approximate error using $S$ in \eqref{eq:Smatrix} to approximate the inverse of $V$
is given in the lemma below, whose proof is the extension of that for Proposition 1 in \cite{yan2016asymptotics}.

\begin{lemma}\label{lemma:appro}
If $V\in{\cal L}_{m,n}(q,Q)$ with $Q/q=o(n)$ and $m/n=O(1)$, then for large enough $n$,
\[
\| V^{-1}-S \| \leq \frac{c_1Q^2}{q^3mn},
\]
where $c_1$ is a constant that dose not depend on $q, Q, m$ and $n$, and
$\| A \|:= \max_{i,j}|a_{i,j}|$ for a general matrix $A=(a_{i,j})$.
\end{lemma}

\begin{proof}
Recall that if $V\in \mathcal{L}_{m,n}(q, Q)$, then for $i=1, \ldots, m+n$,
\begin{eqnarray*}
v_{i,i} & = & \sum_{j=1}^{m+n} (1-\delta_{i,j}) v_{i,j}
=\sum_{j=1}^{m+n} (1-\delta_{j,i}) v_{j,i}\\  &=&  \left\{\begin{array}{ll}
\sum_{j=m+1}^{m+n} v_{i,j}=\sum_{j=m+1}^{m+n} v_{j,i},  &  i=1,\ldots, m, \\
\sum_{j=1}^{m} v_{i,j}=\sum_{j=1}^{m} v_{j,i}, &  i=m+1,\ldots, m+n,
\end{array}
\right.
\end{eqnarray*}
and if $v_{i,j}>0$ for $1\le i \le m,1\le j \le n, i\neq j$, then
\[
q \le v_{i,j} \le Q.
\]
The above equation and inequality will be repeatedly used in this proof.

Let $I$ denote the $(m+n-1)\times (m+n-1)$ identity matrix. Define $F=(f_{i,j})=V^{-1}-S$, $U=(u_{i,j})=I-VS$ and $W=(w_{i,j})=SU$. Then we have the recursion

\begin{equation}\label{recursion}
F=T^{-1}-S=(T^{-1}-S)(I-TS)+S(I-TS)=FU+W.
\end{equation}
Note that for $i=1, \ldots, m, ~j=1, \ldots, m$, we have
\begin{eqnarray*}
\nonumber
u_{i,j}&=&\delta_{i,j}-\sum_{k=1}^{m+n-1} v_{i,k}s_{k,j}\\
       &=&\delta_{i,j}-\left[\sum_{k=1}^{m} v_{i,k}( \frac{\delta_{k,j}}{ v_{j,j} } + \frac{1}{v_{m+n,m+n}})+
        \sum_{k=m+1}^{m+n-1} v_{i,k}( \frac{\delta_{k,j}}{ v_{j,j} }-\frac{1}{v_{m+n, m+n}} )\right]\\
       &=&(\delta_{i,j}-1)\frac{v_{i,j}}{v_{j,j}}- \frac{v_{i,i}-\sum_{k=m+1}^{m+n-1}v_{i,k} }{v_{m+n, m+n}}
       =(\delta_{i,j}-1)\frac{v_{i,j}}{v_{j,j}} - \frac{v_{i,m+n} }{v_{m+n, m+n}}.
\end{eqnarray*}
Similarly, we also have
\[
u_{i,j} =\left\{ \begin{array}{ll}
(\delta_{i,j}-1)\frac{v_{i,j}}{v_{j,j}}+\frac{v_{i,m+n} }{v_{m+n, m+n}} & i=1, \ldots, m; j=m+1, \ldots, m+n-1, \\
(\delta_{i,j}-1)\frac{v_{i,j}}{v_{j,j}} & i=m+1, \ldots, m+n-1; j=1, \ldots, m, m+1, \ldots, m+n-1.
\end{array}
\right.
\]
In all, $u_{i,j}$ can be written in a unified form:
\begin{equation} \label{vij}
u_{i,j}= (\delta_{i,j}-1)\frac{v_{i,j}}{v_{j,j}}+ (1_{\{i\le m, j>m\}}
- 1_{\{i\le m, j\le m\}} )\frac{v_{i,m+n} }{v_{m+n, m+n}},
\end{equation}
where $1_{\{\cdot\}}$ is an indicator function.  Similarly, for $i=1, \ldots, m, ~j=1, \ldots, m$, we have
\begin{eqnarray}
\nonumber
w_{i,j}&=&\sum_{k=1}^{m+n-1} s_{i,k}u_{k,j} \\
\nonumber
&=&
\sum_{k=1}^{m}(\frac{\delta_{i,k}}{v_{i,i}}+\frac{1}{v_{m+n, m+n}} )\left[(\delta_{k,j}-1)\frac{v_{k,j}}{v_{j,j}} - \frac{v_{k,m+n} }{v_{m+n, m+n}} \right]
+ \sum_{k=m+1}^{m+n-1}(-\frac{1}{v_{m+n, m+n}} )\left[(\delta_{k,j}-1)\frac{v_{k,j}}{v_{j,j}}\right ]
\\
\nonumber
&=& \frac{(\delta_{i,j}-1)v_{i,j}}{v_{i,i}v_{j,j}} - \frac{v_{i,m+n}}{v_{i,i}v_{m+n,m+n}} + 0 - \frac{ v_{m+n,m+n}}{v^2_{m+n,m+n}}
+ \frac{ v_{j,j} - v_{m+n,j} }{v_{m+n, m+n}v_{j,j} } \\
\nonumber
&=& \frac{(\delta_{i,j}-1)v_{i,j}}{v_{i,i}v_{j,j}} - \frac{v_{i,m+n}}{v_{i,i}v_{m+n,m+n}} -
\frac{ v_{m+n,j} }{ v_{m+n, m+n} v_{j, j} },
\end{eqnarray}
and
\[
w_{i,j} =\left\{ \begin{array}{ll}
(\delta_{i,j}-1)\frac{v_{i,j}}{v_{i,i}v_{j,j}}+\frac{v_{i,m+n} }{v_{i,i}v_{m+n, m+n}} & i=1, \ldots, m; j=m+1, \ldots, m+n-1, \\
(\delta_{i,j}-1)\frac{v_{i,j}}{v_{i,i}v_{j,j}}+\frac{v_{m+n,j} }{v_{j,j}v_{m+n, m+n}} & i=m+1, \ldots, m+n-1; j=1, \ldots, m, \\
(\delta_{i,j}-1)\frac{v_{i,j}}{v_{i,i}v_{j,j}} & i=m+1, \ldots, m+n-1; j=m+1, \ldots, m+n-1. \\
\end{array}
\right.
\]
Further, when $1\le i\neq j \le m+n $,
\begin{eqnarray*}
&&0\le\frac{v_{i,j}}{v_{i,i}v_{j,j}}\le \frac{Q}{q^2mn},
\end{eqnarray*}
and it is not difficult to show that, when $i,j,k$ are different from each other, we have
\begin{eqnarray*}
|w_{i,i}|&\le & \frac{2Q}{q^2mn},\\
|w_{i,j}|&\le & \frac{3Q}{q^2mn}, \\
|w_{i,j}-w_{i,k}| & \le & \frac{2Q}{q^2mn}, \\
|w_{i,i}-w_{i,k}| & \le & \frac{2Q}{q^2mn}.
\end{eqnarray*}
It follows that
\begin{equation}\label{wij}
\max(|w_{i,j}|, |w_{i,j}-w_{i,k}|)\le \frac{3Q}{q^2mn}~~~~~ \mbox{for all $i,j,k$}.
\end{equation}
Next we use the recursion \eqref{recursion} to obtain an upper bound of the
approximate error $\| F \|$.
By \eqref{recursion} and \eqref{vij}, for any $i\in\{1, \ldots, m+n-1\}$, we have that for $j=1,\ldots,m+n-1$
\begin{equation*}\label{fij}
f_{i,j}=\sum_{k=1}^{m+n-1} f_{i,k}[(\delta_{k,j}-1)\frac{v_{k,j}}{v_{j,j}}+
(1_{\{k\le m, j> m\}} - 1_{\{k\le m, j\le m\}} )\frac{v_{k,m+n}}{v_{m+n,m+n}}] + w_{i,j}.
\end{equation*}
Thus, to prove Lemma \ref{lemma:appro}, it is sufficient to show that $|f_{i,j}|\le c_1Q^2/(q^3mn)$ for any $i,j$.
The condition $m/n=O(1)$ guarantees that $Q/q=o(m)$ if $Q/q=o(n)$.
The remainder of the proof of Lemma \ref{lemma:appro} is similar to the proof of Proposition 1 in \cite{yan2016asymptotics}, and we omit the details here.

\end{proof}

Note that if $Q$ and $q$ are bounded constants, then the upper bound of the above approximation error is on the order of $(mn)^{-1},$ indicating that $S$ is a high-accuracy approximation to $V^{-1}$. Further, based on the above proposition, we immediately have the following lemma.

\begin{lemma}\label{lemma:appro1}
If  $V\in{\cal L}_n(q,Q)$ with $Q/q=o(n)$ and $m/n=O(1)$, then for a vector $\mathbf{x}\in R^{m+n-1}$,
\[
\| V^{-1}\mathbf{x} \|_{\infty} \leq \frac{2c_1Q^2}{q^3mn}+\frac{| x_{m+n} |}{v_{m+n,m+n}}+\max_{i=1,\dots,m+n-1}\frac{| x_{i} |}{v_{i,i}},
\]
where $x_{m+n}:=\sum_{i=1}^{m}x_{i}-\sum_{i=m+1}^{m+n-1}x_{i}$.
\end{lemma}

Similar to Theorem 8 in \cite{yan2016asymptotics}, we have the following lemma for
the rate of convergence for the Newton's iterative sequence to solve
the system of the likelihood equations, whose proof is similar to that in \cite{yan2016asymptotics} and we omit it here.

\begin{lemma}\label{Theorem:newton}
Define a system of equations:
\[
\setlength\arraycolsep{2pt}
\begin{array}{rcl}
F_i({\boldsymbol{\theta}})&=&d_i-\sum_{k=1}^{n}f(\alpha_i+\beta_k),i=1,\dots,m,\\
F_{m+j}({\boldsymbol{\theta}})&=&b_j-\sum_{k=1}^{m}f(\alpha_k+\beta_j),j=1,\dots,n-1,\\
F({\boldsymbol{\theta}})&=&(F_1({\boldsymbol{\theta}}),\dots,F_m({\boldsymbol{\theta}}),F_{m+1}
({\boldsymbol{\theta}}),\dots,F_{m+n-1}({\boldsymbol{\theta}}))^T,
\end{array}
\]
where $f(\cdot)$ is a continuous function with the third derivative. Let $D\subset \mathbb{R}^{m+n-1}$ be a convex set and assume for any $\mathbf{x,y,v}\in D$, we have
\[
\parallel [F^{'}(\mathbf{x})-F^{'}(\mathbf{y})]\mathbf{v}\parallel_{\infty}\leq K_1\parallel\mathbf{x-y}\parallel_{\infty}\parallel\mathbf{v}\parallel_{\infty},
\]
\[
\max_{i=1,\dots,m+n-1}\parallel F^{'}_i(\mathbf{x})-F^{'}_i(\mathbf{y})\parallel_{\infty}\leq K_2\parallel\mathbf{x-y}\parallel_{\infty},
\]
where $F^{'}({\boldsymbol{\theta}})$ is the Jacobian matrix of $F$ on ${\boldsymbol{\theta}}$ and $F^{'}_{i}({\boldsymbol{\theta}})$ is the gradient function of $F_i$ on ${\boldsymbol{\theta}}$.
Consider ${\boldsymbol{\theta}}^{(0)}\in D$ with $\Omega({\boldsymbol{\theta}}^{(0)},2r)\subset D$ where $r=\parallel [F^{'}({\boldsymbol{\theta}}^{(0)})]^{-1}F({\boldsymbol{\theta}}^{(0)})\parallel_{\infty}$ for any ${\boldsymbol{\theta}} \in \Omega({\boldsymbol{\theta}}^{(0)},2r)$.  We assume that $m/n=O(1)$ and
\[
F^{'}({{\boldsymbol{\theta}}})\in {\cal L}_{m,n}(q,Q)~~~~\mbox{or}~~-F^{'}({{\boldsymbol{\theta}}})\in {\cal L}_{m,n}(q,Q).
\]
For $k=1, 2, \ldots$, define the Newton iterates ${\boldsymbol{\theta}}^{(k+1)}={\boldsymbol{\theta}}^{(k)}-[F^{'}({\boldsymbol{\theta}}^{(k)})]^{-1}F({\boldsymbol{\theta}}^{(k)})$.
Let
\[
\rho=\frac{c_1(m+n-1)Q^2K_1}{2q^3mn}+\frac{K_2}{mq}.
\]
If $\rho<1/2,$ then ${\boldsymbol{\theta}}^{(k)} \in \Omega({\boldsymbol{\theta}}^{(0)},2r), k=1, 2, \ldots$, are well defined and satisfy
\[
\parallel{\boldsymbol{\theta}}^{(k+1)}-{\boldsymbol{\theta}}^{(0)}\parallel_{\infty}\leq r/(1-\rho r).
\]
Further, $\lim_{k\rightarrow \infty}{\boldsymbol{\theta}}^{(k)}$ exists and the limiting point is precisely the solution of $F({\boldsymbol{\theta}})=0$ in the rage of ${\boldsymbol{\theta}} \in \Omega({\boldsymbol{\theta}}^{(0)}, 2r)$.
\end{lemma}

\renewcommand\theequation{A\arabic{equation}}
\setcounter{equation}{0}
\subsection*{Appendix A: Proofs for Theorem \ref{Theorem 1}}

We define a system of functions:
\[
\setlength\arraycolsep{2pt}
\begin{array}{rcl}
F_{i}({\boldsymbol{\theta}})&=&
d_i-\sum_{j=1}^{n}\frac{e^{\alpha_{i}+\beta_{j}}} {1+e^{\alpha_{i}+\beta_{j}}}
,~~~i=1,\dots,m,\\
F_{m+j}({\boldsymbol{\theta}})&=&
b_j-\sum_{i=1}^{m}\frac{e^{\alpha_{i}+\beta_{j}}} {1+e^{\alpha_{i}+\beta_{j}}},~~~j=1,\dots,n-1\\
F({\boldsymbol{\theta}})&=&
(F_{1}({{\boldsymbol{\theta}}}),\dots,F_{m+n-1}({{\boldsymbol{\theta}}}))^\top.
\end{array}
\]
Note that the solution to the equation $F({\boldsymbol{\theta}})=0$ is precisely the MLE. Then the Jacobin matrix $F^{'}({\boldsymbol{\theta}})$ of $F({\boldsymbol{\theta}})$ can
be calculated as follows. For $i=1,\dots,m,$
\[
\frac{\partial{F_i}}{\partial{{\alpha_l}}}=0,l=1,\dots,m,l\neq i;~~
\frac{\partial{F_i}}{\partial{{\alpha_i}}}=-\sum^{n}_{j=1}\frac{e^{\alpha_{i}+\beta_{j}}}{(1+e^{\alpha_{i}+\beta_{j}})^2},
\]
\[
\frac{\partial{F_i}}{\partial{ {\beta_j}}}=-\frac{e^{\alpha_{i}+\beta_{j}}}{(1+e^{\alpha_{i}+\beta_{j}})^2},~j=1,\dots,n-1,
\]
and for $j=1,\dots,n-1$,
\[
\frac{\partial{F_{m+j}}}{\partial{ {\alpha_l}}}=-\frac{e^{\alpha_{l}+\beta_{j}}}{(1+e^{\alpha_{l}+\beta_{j}})^2},~l=1,\dots,m,
\]
\[
\frac{\partial{F_{m+j}}}{\partial{{\beta_j}}}=
-\sum^{m}_{i=1}\frac{e^{\alpha_{i}+\beta_{j}}}{(1+e^{\alpha_{i}+\beta_{j}})^2};~~\frac{\partial{F_{m+j}}}{\partial{{\beta_k}}}=0
,~k=1,\dots,n-1,k\neq j.
\]
Since $e^{x}/(1+e^{x})^2$ is a decreasing function on $x$ when $x\geq 0$ and an increasing function when $x\leq 0$.
Consequently, for any $i, j$,  we have
\[
\frac{e^{2\parallel{{\boldsymbol{\theta}}}\parallel_{\infty}}}{(1+e^{2\parallel{{\boldsymbol{\theta}}}\parallel_{\infty}})^2}
\leq-F^{'}_{i,j}({\boldsymbol{\theta}})\leq\frac{1}{4}.
\]
According to the definition of $ {\cal L}_{m,n}(q, Q)$,
we have that  $-F^{'}({\boldsymbol{\theta}})\in {\cal L}_{m,n}(q,Q)$, where
\[
q=\frac{e^{2\parallel{{\boldsymbol{\theta}}}\parallel_{\infty}}}{(1+e^{2\parallel{{\boldsymbol{\theta}}}\parallel_{\infty}})^2},
~~Q=\frac{1}{4}.
\]
Therefore, Lemma \ref{Theorem:newton} and Lemma \ref{lemma:appro} can be applied.
Let ${{\boldsymbol{\theta}}}^*$ denote the true parameter vector.
The constants $K_1, K_2$ and $r$ in the upper bounds of Lemma \ref{Theorem:newton} are given in the following lemma.

\begin{lemma}\label{lemma 2}
Take $D=R^{m+n-1}$ and ${{\boldsymbol{\theta}}}^0={{\boldsymbol{\theta}}}^*$ in Lemma \ref{Theorem:newton}.
Assume
\begin{equation}\label{B1}
\max\{\max_{i=1,\dots,m}|d_i-\mathbb{E}(d_i)|, \max_{j=1,\dots,n}|b_j-\mathbb{E}(b_j)|\}\leq \sqrt{n\log{n}}.
\end{equation}
If $m/n=O(1)$, then we can choose the constants $K_1, K_2$ and $r$ in Lemma \ref{Theorem:newton} as
\[
K_1=n,
K_2=\frac{n}{2},
r\leq \frac{(\log n)^{1/2}}{n^{1/2}}\left(c_{11}e^{6\parallel{\boldsymbol{\theta}}^{*}\parallel_{\infty}}+c_{12}
e^{2\parallel{\boldsymbol{\theta}}^{*}\parallel_{\infty}}\right),
\]
where $c_{11}, c_{12}$ are constants.
\end{lemma}
\begin{proof}
For fixed $m,n$, we first derive $K_1$ and $K_2$ in the inequalities of Lemma \ref{Theorem:newton}. Let $\mathbf{x},\mathbf{y}\in R^{m+n-1}$ and
\[{F_i}^{'}({\boldsymbol{\theta}})=({F^{'}_{i,1}}({\boldsymbol{\theta}}),\dots,{F^{'}_{i,m+n-1}}({\boldsymbol{\theta}}):=
(\frac{\partial{F_i}}{\partial{\alpha_1}},\dots, \frac{\partial{F_i}}{\partial{\alpha_m}},\frac{\partial{F_i}}{\partial{\beta_1}},\dots,
\frac{\partial{F_i}}{\partial{\beta_{n-1}}}).\]
Then, for $i=1,\dots,m$, we have
\[\frac{\partial^2{F_i}}{\partial{{\alpha_l}}\partial{ {\alpha_s}}}=0,s\neq l;-\sum^{n}_{j=1}\frac{e^{\alpha_{i}+\beta_{j}}(1-e^{\alpha_{i}+\beta_{j}})}{(1+e^{\alpha_{i}+\beta_{j}})^3},\]

\[\frac{\partial^2{F_i}}{\partial{ {\alpha_i}}\partial{ {\beta_s}}}=-\frac{e^{\alpha_{i}+\beta_{j}}(1-e^{\alpha_{i}+\beta_{j}})}{(1+e^{\alpha_{i}+\beta_{j}})^3},s=1,\dots,n-1,s\neq i;\frac{\partial^2{F_i}}{\partial{ {\alpha_i}}\partial{ {\beta_i}}}=0,\]
\[\frac{\partial^2{F_i}}{\partial{{\beta_j}}^2}=-\frac{e^{\alpha_{i}+\beta_{j}}(1-e^{\alpha_{i}+\beta_{j}})}{(1+e^{\alpha_{i}+\beta_{j}})^3},
j=1,\dots,n-1;\frac{\partial^2{F_i}}{\partial{ {\beta_s}}\partial{ {\beta_l}}}=0,s\neq l.\]
Note that
\begin{equation}\label{B2}
\mid\frac{e^{\alpha_{i}+\beta_{j}}(1-e^{\alpha_{i}+\beta_{j}})}{(1+e^{\alpha_{i}+\beta_{j}})^3}
\mid\leq\frac{e^{\alpha_{i}+\beta_{j}}}{(1+e^{\alpha_{i}+\beta_{j}})^2}\leq\frac{1}{4}.
\end{equation}
By the mean value theorem for vector-valued functions (\cite{lang1993real}, p.341), we have
\[F^{'}_i(\mathbf{x})-F^{'}_i(\mathbf{y})=J^{(i)}(\mathbf{x}-\mathbf{y}),\]
where
\[J^{(i)}_{s,l}=\int^1_0\frac{\partial F^{'}_{i,s}}{\partial{{\boldsymbol{\theta}}}_l}(t\mathbf{x}+(1-t)\mathbf{y})dt,~~s,l=1,\dots,m+n-1.\]
Therefore,
\[\max_s\sum_{l}^{m+n-1}\mid J^{(i)}\mid\leq \frac{n}{2}, ~~ \sum^{}_{s,l}\mid J^{(i)}_{(s,l)}\mid\leq n\]
Similarly, for $i=m+1, \dots, m+n-1,$ we also have $F^{'}_i(\mathbf{x})-F^{'}_i(\mathbf{y})=J^{(i)}(\mathbf{x}-\mathbf{y})$ and $\sum^{}_{s,l}
| J^{(i)}_{(s,l)} |\leq m$.
Consequently,
\[
\parallel F^{'}_i(\mathbf{x})-F^{'}_i(\mathbf{y})\parallel_{\infty}\leq \parallel J^{(i)}\parallel_{\infty}\parallel \mathbf{x-y}\parallel_{\infty}\leq \frac{n}{2}\parallel \mathbf{x-y}\parallel_{\infty},i=1,\dots,m+n-1,
\]
and for $\forall~\mathbf{v}\in R^{m+n-1}$,
\[
\setlength\arraycolsep{2pt}
\begin{array}{rcl}
\parallel [F^{'}_i(\mathbf{x})-F^{'}_i(\mathbf{y})]\mathbf{v}\parallel_{\infty}&=&\displaystyle\max_{i}\mid \sum_{j=1}^{m+n-1}(F_{i,j}^{'}(\mathbf{x})-F_{i,j}^{'}(\mathbf{y}))v_j\mid\\
&=&\displaystyle\max_{i}\mid(\mathbf{x-y})J^{(i)}\mathbf{v}\mid\\
&\leq&\parallel\mathbf{x-y}\parallel_{\infty}\parallel\mathbf{v}\parallel_{\infty}\sum^{}_{k,j}\mid J^{(i)}_{(s,l)}\mid\\
&\leq& n \parallel\mathbf{x-y}\parallel_{\infty}\parallel\mathbf{v}\parallel_{\infty}.
\end{array}
\]
So we can choose $K_1=n$ and $K_2=n/2$ in Lemma \ref{Theorem:newton}.

It is obvious that $-F^{'}({ {\boldsymbol{\theta}}}^{*})\in {\cal L}_n(q_{*},Q_{*})$, where
\[
q_{*}=\frac{e^{2\parallel{{\boldsymbol{\theta}}^{*}}\parallel_{\infty}}}{(1+e^{2\parallel{{\boldsymbol{\theta}}^{*}}\parallel_{\infty}})^2},
~~Q_{*}=\frac{1}{4}.
\]
Note that
\[F( {\boldsymbol{\theta}}^*)=(d_1-\mathbb{E}(d_1),\dots,d_m-\mathbb{E}(d_m),b_1-\mathbb{E}(b_1),\dots,b_{n-1}-\mathbb{E}(b_{n-1})).\]
By the assumption of \ref{B1} and Lemma \ref{lemma:appro1}, if $m/n=O(1)$, then we have
\[
\setlength\arraycolsep{2pt}
\begin{array}{rcl}
r=\parallel[F^{'}(
 {\boldsymbol{\theta}}^*)]^{-1}F( {\boldsymbol{\theta}}^*)\parallel_{\infty}&\leq&\frac{2c_1(m+n-1)Q_{*}^2\parallel F( {\boldsymbol{\theta}}^*)\parallel_{\infty}}{q_*^3mn}+\displaystyle\max_{i=1,\dots,m+n-1}\frac{\mid F_{i}( {\boldsymbol{\theta}}^*)\mid}{v_{i,i}}+\frac{\mid F_{m+n}( {\boldsymbol{\theta}}^*)\mid}{v_{m+n,m+n}}\\
&\leq&\frac{(\log n)^{1/2}}{n^{1/2}}\left(c_{11}e^{6\parallel{\boldsymbol{\theta}}^{*}\parallel_{\infty}}+c_{12}
e^{2\parallel{\boldsymbol{\theta}}^{*}\parallel_{\infty}}\right),
\end{array}
\]
where $c_{11},c_{12}$ are constants.
\end{proof}

The following lemma assures that condition (\ref{B1}) holds with a large probability.
\begin{lemma}\label{lemma 3}
With probability at least $1-4/n$, we have
\[\max\{\max_{i=1,\dots,m}|d_i-\mathbb{E}(d_i)|,\max_{j=1,\dots,n}|b_j-\mathbb{E}(b_j)|\}\leq \sqrt{n\log(n)}.\]
\end{lemma}

\begin{proof}
Note that $m<n$.
By \cite{hoeffding1963probability}'s  inequality, we have
\[
P\left(|d_i-\mathbb{E}(d_i)|\geq\sqrt{n\log{n}}\right)\leq 2\exp{\{-\frac{2n\log{n}}{m}\}}\leq\frac{2}{n^{2n/m}} \leq\frac{2}{n^2}.
\]
Therefore,
\[
\setlength\arraycolsep{0pt}
\begin{array}{rcl}
P\left(\displaystyle\max_{i}|d_i-\mathbb{E}(d_i)|\geq\sqrt{n\log{n}}\right)&\leq&P\left(\displaystyle
\bigcup_{i}|d_i-\mathbb{E}(d_i)|\geq\sqrt{n\log{n}}\right)\\
&\leq&\displaystyle\sum_{i=1}^{m}P\left(|d_i-\mathbb{E}(d_i)|\geq\sqrt{n\log{n}}\right)\\
&\leq&m\times\frac{2}{n^2}\leq\frac{2}{n}.
\end{array}
\]
Similarly, we have
\[P\left(\displaystyle\max_{j}|b_j-\mathbb{E}(b_j)|\geq\sqrt{n\log{n}}\right)\leq\frac{2}{n}.\]
Consequently,
\[\setlength\arraycolsep{0pt}
\small
\begin{array}{rcl}
&~&P\left(\max\{\max_{i=1,\dots,m}|d_i-\mathbb{E}(d_i)|,\max_{j=1,\dots,n}|b_j-\mathbb{E}(b_j)|\}\geq \sqrt{n\log{n}}\right)\\
&\leq&P\left(\displaystyle\max_{i}|d_i-\mathbb{E}(d_i)|\geq \sqrt{n\log{n}}\right)+P\left(\displaystyle\max_{j}|b_j-\mathbb{E}(b_j)|\geq \sqrt{n\log{n}}\right)\\
&\leq&\frac{4}{n}.
\end{array}\]
This is equivalent to Lemma \ref{lemma 3}.
\end{proof}

\begin{proof}[\bf Proof of Theorem \ref{Theorem 1}]
Assume that condition (\ref{B1}) holds.
Recall that the Newton's iterates in Lemma \ref{Theorem:newton},
$ {\boldsymbol{\theta}}^{(k+1)}=[F^{'}( {\boldsymbol{\theta}}^{(k)})]^{-1}F( {\boldsymbol{\theta}}^{(k)})$ with $ {\boldsymbol{\theta}}^{(0)}= {\boldsymbol{\theta}}^{*}.$ If $ {\boldsymbol{\theta}}\in\Omega( {\boldsymbol{\theta}}^{*},2r)$, then $-F^{'}({ {\boldsymbol{\theta}}}^{*})\in {\cal L}_{m,n}(q,Q)$ with
\[
q=\frac{e^{2\parallel{{\boldsymbol{\theta}}^{*}}\parallel_{\infty}+2r}}{(1+e^{2\parallel{{\boldsymbol{\theta}}^{*}}\parallel_{\infty}+2r})^2},
~~Q=\frac{1}{4}.
\]
By Lemma \ref{lemma 2} and condition(\ref{B1}), for sufficient small $r$,
\[
\setlength\arraycolsep{2pt}
\begin{array}{rcl}
\rho r&\leq&\left[\frac{c_1(m+n-1)Q^2n}{2q^3mn}+\frac{n}{2nq}\right]\times\frac{(\log n)^{1/2}}{n^{1/2}}\left(c_{11}e^{6\parallel{\boldsymbol{\theta}}^{*}\parallel_{\infty}}+c_{12}
e^{2\parallel{\boldsymbol{\theta}}^{*}\parallel_{\infty}}\right)\\
&\leq&O(\frac{(\log n)^{1/2}e^{12\parallel{\boldsymbol{\theta}}^{*}\parallel_{\infty}}}{n^{1/2}})+O(\frac{(\log n)^{1/2}e^{8\parallel{\boldsymbol{\theta}}^{*}\parallel_{\infty}}}{n^{1/2}}).
\end{array}
\]
If ${{\boldsymbol{\theta}}}^*\in \mathbb{R}^{m+n-1}$ with $\| {{\boldsymbol{\theta}}}^*\|_{\infty}\leq \tau\log n$, where $0<\tau<1/24$ is a constant, then as $n\rightarrow\infty$,
\[
(\log n)^{1/2}n^{-1/2}e^{12\parallel{\boldsymbol{\theta}}^{*}\parallel_{\infty}}\leq (\log n)^{1/2}n^{-1/2+12\tau}\rightarrow 0.
\]
Therefore, $\rho r\rightarrow 0$ as $n\rightarrow\infty$.
By Lemma \ref{Theorem:newton}, $\lim_{n\rightarrow\infty}{\widehat{{\boldsymbol{\theta}}}}^{(n)}$ exists. Denote the limit as $\widehat{\boldsymbol{\theta}}$.
Then it satisfies
\[
\|{\widehat{{\boldsymbol{\theta}}}}- {\boldsymbol{\theta}}^*\|_{\infty}\leq2r=O(\frac{(\log n)^{1/2}e^{6\parallel{\boldsymbol{\theta}}^{*}\parallel_{\infty}}}{n^{1/2}})=o(1).
\]
By Lemma \ref{lemma 3}, condition (\ref{B1}) holds with probability approaching to one, thus the above inequality also holds with probability approaching to one. The uniqueness of the MLE comes from Proposition 5 in
\cite{yan2016asymptotics}.
\end{proof}

\subsection*{Appendix B: Proofs for Theorem \ref{Theorem:central}}
\renewcommand\theequation{C\arabic{equation}}
\setcounter{equation}{0}  
We first present one proposition.
Since $d_i=\sum^{}_{k}a_{i,k}$ and $b_j=\sum^{}_{k}a_{k,j}$ are sums of $m$ and $n$ independent random variables,
by the central limit theorem for the bounded case in Lo\`{e}ve (\cite{Loeve(1977)}, page 289), we know that ${v_{i,i}}^{-1/2}(d_i-\mathbb{E}(d_i))$ and ${v_{m+j,m+j}}^{-1/2}(b_j-\mathbb{E}(b_j))$ are asymptotically standard normal if $v_{i,i}$ and $v_{m+j,m+j}$ diverge, respectively.
Note that
\[
\frac{m e^{2\parallel{{\boldsymbol{\theta}}^{*}}\parallel_{\infty}}}{(1+e^{2\parallel {{\boldsymbol{\theta}}}^*\parallel_{\infty}})^2}\leq v_{i,i}\leq \frac{m}{4},~~\frac{n e^{2\parallel{{\boldsymbol{\theta}}^{*}}\parallel_{\infty}}}{(1+e^{2\parallel {{\boldsymbol{\theta}}}^*\parallel_{\infty}})^2}\leq v_{m+j,m+j}\leq \frac{n}{4}.
\]
Then we have the following proposition.

\begin{proposition}\label{proposition 2}
Assume that $X\sim \mathbb{P}_{{{\boldsymbol{\theta}}}^*}$.
If $e^{\parallel {{\boldsymbol{\theta}}}^*\parallel_{\infty}}=o(n^{1/2})$, then for any fixed $k\geq 1$, as $n\rightarrow \infty$, the vector consisting of the first $k$ elements of $S\{\mathbf{g}-\mathbb{E}(\mathbf{g})\}$ is asymptotically multivariate normal with mean zero and covariance matrix given by the upper left $k\times k$ block of $S$.
\end{proposition}

To complete the proof of Theorem \ref{Theorem:central}, we need two lemmas as follows.
\begin{lemma}\label{lemma 6}
Let $R=V^{-1}-S$ and $U=Cov[R\{\mathbf{g}-\mathbb{E}\mathbf{g}\}]$. Then
\begin{equation*}
\| U \| \leq \| V^{-1}-S\|+\frac{3(1+e^{2\parallel{{\boldsymbol{\theta}}^{*}}\parallel_{\infty}})^4}{4mne^{4\parallel
{{\boldsymbol{\theta}}^{*}}\parallel_{\infty}}}.
\end{equation*}
\end{lemma}

\begin{proof}
Note that
\[U=RVR^T=(V^{-1}-S)V(V^{-1}-S)^T=(V^{-1}-S)-S(I-VS),\]
where $I$ is a $(m+n-1)\times(m+n-1)$ diagonal matrix, and by \eqref{wij}, we have
\[\mid \{S(I-VS)\}_{i,j}\mid=\mid w_{i,j}\mid\leq\frac{3(1+e^{2\parallel{{\boldsymbol{\theta}}^{*}}\parallel_{\infty}})^4}{4mne^{4\parallel{{\boldsymbol
{\theta}}^{*}}\parallel_{\infty}}}.\]
Thus,
\[\parallel U \parallel\leq\parallel V^{-1}-S\parallel+\parallel \{S(I_{m+n-1}-VS)\}\parallel\leq\parallel V^{-1}-S\parallel+\frac{3(1+e^{2\parallel{{\boldsymbol{\theta}}^{*}}\parallel_{\infty}})^4}{4mne^{4\parallel
{{\boldsymbol{\theta}}^{*}}\parallel_{\infty}}}.\]
\end{proof}

\begin{lemma}\label{Lemma 7}
Assume that the conditions in Theorem \ref{Theorem 1} hold.
If $\| {{\boldsymbol{\theta}}}^*\|_{\infty}\leq \tau\log n$ with $\tau<1/24$ and $m/n=O(1)$, then for any $i$,
\[
\widehat{\boldsymbol{\theta}}_{i}- {\boldsymbol{\theta}}^{*}_{i}=[V^{-1}\{\mathbf{g}-\mathbb{E}\mathbf{g}\}]_{i}+o_p(n^{-1/2}).
\]
\end{lemma}

\begin{proof}
By Theorem \ref{Theorem 1}, if $m/n=O(1)$, then we have
\[
\widehat{\rho}_n:=\max_{1\leq i\leq m+n-1}\mid\widehat{\boldsymbol{\theta}}_{i}- {\boldsymbol{\theta}}^{*}_{i}\mid=O_p(\frac{(\log n)^{1/2}e^{6\parallel{\boldsymbol{\theta}}^{*}\parallel_{\infty}}}{n^{1/2}}).
\]
Let $\widehat\gamma_{i,j}=\widehat{\alpha}_{i}+\widehat{\beta}_{j}- \alpha_{i}^{*}- \beta_{j}^{*}$.
By the Taylor expansion, for any $1\leq i \leq m,1\leq j \leq n, i\neq j$,
\[
\frac{e^{\widehat{\alpha}_{i}+\widehat{\beta}_{j}}}{1+e^{\widehat{\alpha}_{i}+\widehat{\beta}_{j}}}
-\frac{e^{{ \alpha}_{i}^*+{\beta}_{j}^*}}{1+e^{{ \alpha}_{i}^*+
 {\beta}_{j}^*}}=\frac{e^{{ \alpha}_{i}^*+ {\beta}_{j}^*}}{(1+e^{{ \alpha}_{i}^*+ {\beta}_{j}^*})^2}\hat\gamma_{i,j}+h_{i,j},
\]
where
\[
h_{ij}=\frac{e^{ {\alpha}_{i}^*+ {\beta}_{j}^*+\phi_{i,j}\widehat\gamma_{i,j}}(1-e^{ {\alpha}_{i}^*+ {\beta}_{j}^*+\phi_{i,j}\widehat\gamma_{i,j}})}{2(e^{ {\alpha}_{i}^*+ {\beta}_{j}^*+\phi_{i,j}\widehat\gamma_{i,j}})^3}\widehat\gamma_{i,j}^2,\]
and $0\leq\phi_{i,j}\leq 1$. By the likehood equations (\ref{eq:2.2}), it is not difficult to verify that
\[
\mathbf{g}-\mathbb{E}\mathbf{g}=V(\widehat{\boldsymbol{\theta}}- {\boldsymbol{\theta}}^*)+\mathbf{h},
\]
where $\mathbf{h}=(h_1,\dots,h_{m+n-1})^T$ and
\[
h_i=\sum_{k=1}^{n}h_{i,k}, i=1, \dots, m, ~~ h_{m+i}=\sum_{k=1}^{m}h_{k,i}, i=1, \dots, n-1.
\]
Equivalently,
\begin{equation}\label{C2}
\widehat{\boldsymbol{\theta}}- {\boldsymbol{\theta}}^*=V^{-1}(\mathbf{g}-\mathbb{E}\mathbf{g})+V^{-1}\mathbf{h}.
\end{equation}
By \eqref{B2}, it is easy to show
\[
\mid h_{i,j}\mid\leq\mid\widehat\gamma_{i,j}^2/2 \mid\leq2\widehat{\rho}_n^2,~~~~\mid h_{i}\mid\leq\sum_{i, j}\mid h_{i,j}\mid
\leq 2n\widehat{\rho}_n^2.
\]
Because
\[
h_{m+n}=\sum_{i=1}^{m}h_i-\sum_{j=1}^{n-1}h_{m+j}=\sum_{i=1}^{m}\sum_{j=1}^{n}h_{i,j}-
\sum_{j=1}^{n-1}\sum_{i=1}^{m}h_{i,j}=\sum_{i=1}^{m}h_{i,n}.
\]
Therefore
\[
\mid h_{m+n}\mid\leq (m+n)\widehat{\rho}_n^2.
\]
Note that $(S\mathbf{h})_i=h_i/v_{i,i}+(-1)^{1_{\{i>m\}}}h_{m+n}/v_{m+n,m+n},$ and $(V^{-1}\mathbf{h})_i=(S\mathbf{h})_i+(R\mathbf{h})_i$. Then we have
\[\mid(S\mathbf{h})_i\mid\leq\frac{\mid h_i\mid}{v_{i,i}}+\frac{\mid h_{2n}\mid}{v_{m+n,m+n}}\leq\frac{{4\widehat{\rho}_n^2\cdot(1+e^{2\parallel{{\boldsymbol
{\theta}}^{*}}\parallel_{\infty}})^2}}{e^{2\parallel{{\boldsymbol{\theta}}^{*}}\parallel_{\infty}}}\leq O \left(\frac{ e^{14\parallel{\boldsymbol{\theta}}^{*}\parallel_{\infty}}\log n}{n}\right),\]
by Lemma \ref{lemma:appro}, we have
\[\mid(R\mathbf{h})_i\mid\leq\| R\|_{\infty}\times[(m+n-1)\max_{i}| h_i|]\leq O \left(\frac{e^{18\parallel{\boldsymbol{\theta}}^{*}\parallel_{\infty}}\log n}{n}\right).\]
If $\parallel {{\boldsymbol{\theta}}}^*\parallel_{\infty}\leq \tau\log n$, and $\tau<1/36$, then $\mid(V^{-1}\mathbf{h})_i\mid\leq\mid(S\mathbf{h})_i\mid+\mid(R\mathbf{h})_i\mid=o(n^{-1/2})$.

\end{proof}

\begin{proof}[\bf Proof of Theorem \ref{Theorem:central}]
By \eqref{C2}, we have
\[
(\widehat{\boldsymbol{\theta}}-{\boldsymbol{\theta}}^{*})_{i}=[S\{\mathbf{g}-\mathbb{E}
\mathbf{g}\}]_i+[R\{\mathbf{g}-\mathbb{E}\mathbf{g}\}]_i+(V^{-1}\mathbf{h})_i.
\]
By Lemma \ref{lemma 6}, we have
\[
\parallel U \parallel\leq\| V^{-1}-S\|+\frac{3(1+e^{2\parallel{{\boldsymbol{\theta}}^{*}}\parallel_{\infty}})^4}{4mne^{4\parallel
{{\boldsymbol{\theta}}^{*}}\parallel_{\infty}}}\leq  O(\frac{e^{6\parallel{\boldsymbol{\theta}}^{*}\parallel_{\infty}}}{mn})+O(\frac{e^{4\parallel
{\boldsymbol{\theta}}^{*}\parallel_{\infty}}}{mn})= O(\frac{e^{6\parallel{\boldsymbol{\theta}}^{*}\parallel_{\infty}}}{mn}).
\]
By Chebyshev's inequality, if $\parallel {{\boldsymbol{\theta}}}^*\parallel_{\infty}\leq \tau\log n$, $m/n=O(1)$, and $\tau<1/36$, then
\[
P\left\{\frac{[R\{\mathbf{g}-\mathbb{E}\mathbf{g}\}]_i}{n^{-1/2}}>\epsilon\right\}\leq\frac{Cov[R\{\mathbf{g}-
\mathbb{E}\mathbf{g}\}]_i}{n\epsilon^2}\leq\frac{1}{n\epsilon^2}O(\frac{e^{6\parallel
{\boldsymbol{\theta}}^{*}\parallel_{\infty}}}{mn})=o(1),
\]
Therefore, we have
\[
[R\{\mathbf{g}-\mathbb{E}\mathbf{g}\}]_i=o_p(n^{-1/2}).
\]
By  Lemma \ref{lemma 6} and Lemma \ref{Lemma 7}, we have
\[
(\widehat{\boldsymbol{\theta}}-{\boldsymbol{\theta}}^{*})_{i}=[S\{\mathbf{g}-\mathbb{E}\mathbf{g}\}]_i+o_p(n^{-1/2}).
\]
Theorem \ref{Theorem:central} follows directly from Proposition \ref{proposition 2}.
\end{proof}

\section*{Acknowledgements}
We are very grateful to one anonymous referee and the Editor for their valuable comments that have greatly improved the manuscript.
Qin's research is partially supported by the National Natural Science Foundation of China (No.11271147, 11471135).
Yan's research is partially supported by the National Natural Science Foundation of China (No.11401239) and the self-determined research funds of CCNU from the colleges's basic research and operation of MOE (CCNU15A02032, CCNU15ZD011) and a fund from KLAS (130026507). Zhang's research is supported by a fund from CCNU (2016CXZZ157).

\bibliography{reference}

\begin{thebibliography}{}

\bibitem[\protect\astroncite{Bhattacharyya and
  Bickel}{2016}]{Bhattacharyya2016}
Bhattacharyya, S. and Bickel, P.~J. (2016).
\newblock Spectral clustering and block models: A review and a new algorithm.
\newblock {\em Statistical Analysis for High-Dimensional Data}, pages 67--90.

\bibitem[\protect\astroncite{Chatterjee and
  Diaconis}{2013}]{Chatterjee:Diaconis:2013}
Chatterjee, S. and Diaconis, P. (2013).
\newblock Estimating and understanding exponential random graph models.
\newblock {\em The Annals of Statistics}, 41(5):2428--2461.

\bibitem[\protect\astroncite{Chatterjee et~al.}{2011}]{chatterjee2011}
Chatterjee, S., Diaconis, P., and Sly, A. (2011).
\newblock Random graphs with a given degree sequence.
\newblock {\em Ann. Appl. Probab.}, 21(4):1400--1435.

\bibitem[\protect\astroncite{Conyon and Muldoon}{2004}]{conyon2004small}
Conyon, M.~J. and Muldoon, M.~R. (2004).
\newblock The small world network structure of boards of directors.
\newblock {\em Available at SSRN 546963,
  \url{http://ssrn.com/abstract=546963}}.

\bibitem[\protect\astroncite{Haberman}{1977}]{Haberman:1977}
Haberman, S.~J. (1977).
\newblock Maximum likelihood estimates in exponential response models.
\newblock {\em The Annals of Statistics}, 5:815--841.

\bibitem[\protect\astroncite{Hillar and Wibisono}{2013}]{hillar2013maximum}
Hillar, C. and Wibisono, A. (2013).
\newblock Maximum entropy distributions on graphs.
\newblock {\em \em Avaible at: \url{http://arxiv.org/abs/1301.3321}}.

\bibitem[\protect\astroncite{Hoeffding}{1963}]{hoeffding1963probability}
Hoeffding, W. (1963).
\newblock Probability inequalities for sums of bounded random variables.
\newblock {\em Journal of the American statistical association},
  58(301):13--30.

\bibitem[\protect\astroncite{Holl and Leinhardt}{1981}]{Holland:Leinhardt:1981}
Holl, P.~W. and Leinhardt, S. (1981).
\newblock An exponential family of probability distributions for directed
  graphs (with discussion).
\newblock {\em Journal of the American Statistical Association},
  76(373):33--65.

\bibitem[\protect\astroncite{Iacobucci and
  Wasserman}{1990}]{iacobucci1990social}
Iacobucci, D. and Wasserman, S. (1990).
\newblock Social networks with two sets of actors.
\newblock {\em Psychometrika}, 55(4):707--720.

\bibitem[\protect\astroncite{Lang}{1993}]{lang1993real}
Lang, S. (1993).
\newblock Real and functional analysis.
\newblock {\em New York:Springer-Verlag}.

\bibitem[\protect\astroncite{Latapy et~al.}{2008}]{latapy2008basic}
Latapy, M., Magnien, C., and Del~Vecchio, N. (2008).
\newblock Basic notions for the analysis of large two-mode networks.
\newblock {\em Social networks}, 30(1):31--48.

\bibitem[\protect\astroncite{Loeve}{1977}]{Loeve(1977)}
Loeve, M. (1977).
\newblock Probability theory. 4th ed.
\newblock {\em New York:Springer-Verlag}.

\bibitem[\protect\astroncite{Rasch}{1960}]{Rasch:1960}
Rasch, G. (1960).
\newblock Probabilistic models for some intelligence and attainment tests.
\newblock {\em Copenhagen: Paedagogiske Institut.}

\bibitem[\protect\astroncite{Robins and Alexander}{2004}]{robins2004small}
Robins, G. and Alexander, M. (2004).
\newblock Small worlds among interlocking directors: Network structure and
  distance in bipartite graphs.
\newblock {\em Computational \& Mathematical Organization Theory},
  10(1):69--94.

\bibitem[\protect\astroncite{Shalizi and Rinaldo}{2013}]{Shalizi:Rinaldo:2013}
Shalizi, C.~R. and Rinaldo, A. (2013).
\newblock Consistency under sampling of exponential random graph models.
\newblock {\em The Annals of Statistics}, 41(2):508--535.

\bibitem[\protect\astroncite{Snijders et~al.}{2013}]{snijders2013model}
Snijders, T.~A., Lomi, A., and Torl{\'o}, V.~J. (2013).
\newblock A model for the multiplex dynamics of two-mode and one-mode networks,
  with an application to employment preference, friendship, and advice.
\newblock {\em Social networks}, 35(2):265--276.

\bibitem[\protect\astroncite{Wang et~al.}{2009}]{wang2009exponential}
Wang, P., Sharpe, K., Robins, G.~L., and Pattison, P.~E. (2009).
\newblock Exponential random graph (p*) models for affiliation networks.
\newblock {\em Social Networks}, 31(1):12--25.

\bibitem[\protect\astroncite{Yan et~al.}{2016a}]{yan2016asymptotics}
Yan, T., Leng, C., Zhu, J., et~al. (2016a).
\newblock Asymptotics in directed exponential random graph models with an
  increasing bi-degree sequence.
\newblock {\em The Annals of Statistics}, 44(1):31--57.

\bibitem[\protect\astroncite{Yan et~al.}{2016b}]{Yan:Qin:Wang:2016}
Yan, T., Qin, H., and Wang, H. (2016b).
\newblock Asymptotics in undirected random graph models parameterized by the
  strengths of vertices.
\newblock {\em Statistica Sinica}, (26):273--293.

\bibitem[\protect\astroncite{Yan and Xu}{2013}]{yan2013central}
Yan, T. and Xu, J. (2013).
\newblock A central limit theorem in the $\beta$-model for undirected random
  graphs with a diverging number of vertices.
\newblock {\em Biometrika}, 100(2):519--524.

\end{thebibliography}
\bibliographystyle{apa}

\newpage
\begin{figure}[H]
\centering
\caption{The QQ plots of $\hat{v}_{ii}^{-1/2}(\hat{\alpha}_i-\alpha_i)$ ($m=100, n=200$). }
\label{figure-discrete-qq,alpha}
\includegraphics[height=4in, width=6in, angle=0]{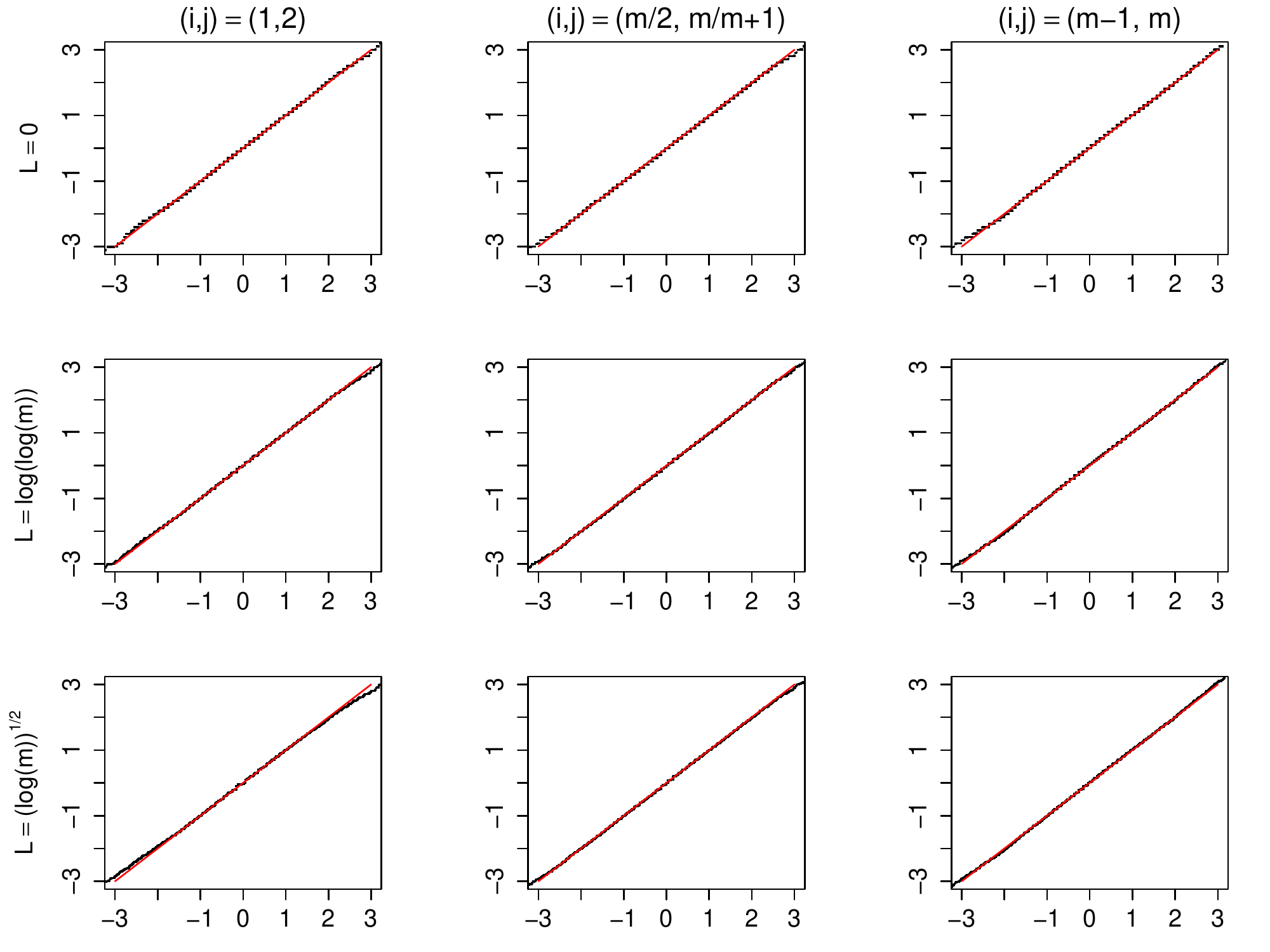}
\end{figure}

\begin{figure}[H]
\centering
\caption{The QQ plots of $\hat{v}_{jj}^{-1/2}(\hat{\beta}_j-\beta_j)$ ($m=100, n=200$). }
\label{figure-discrete-qq,beta}
\includegraphics[height=4in, width=6in, angle=0]{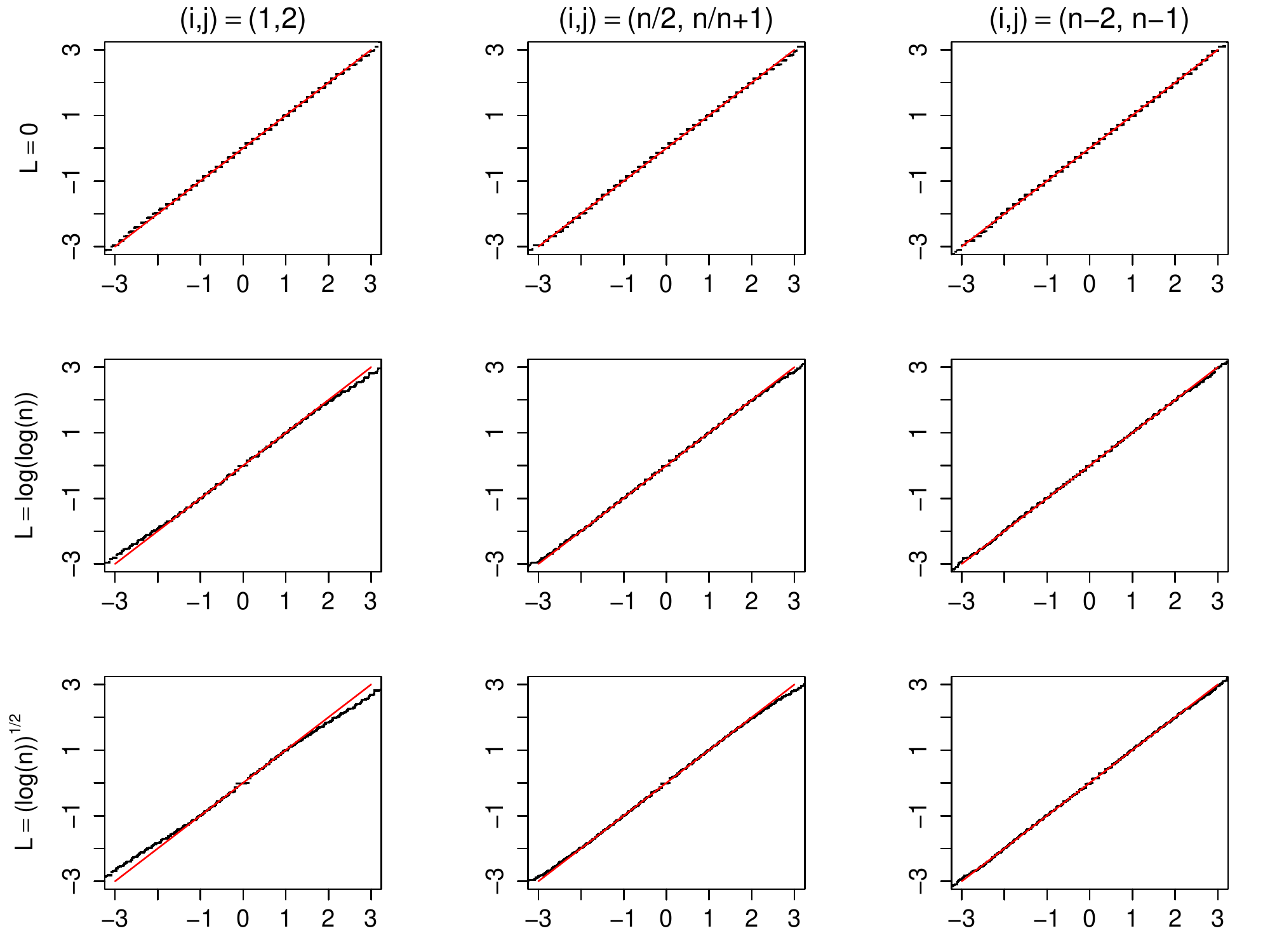}
\end{figure}
\begin{table}[H]\centering
\caption{Estimated coverage probabilities of $\alpha_i-\alpha_j$ for pair $(i,j)$ as well as the length of confidence intervals (in square brackets), and the probabilities that the MLE does not exist (in parentheses), multiplied by $100$.}
\label{Table 1:Estimatation of alpha}
\small
\vskip5pt
\begin{tabular}{ccccccc}
\\
\hline
m       &  $(i,j)$ & $L=0$ & $L=\log m$ & $L=(\log m)^{1/2}$ & $L=\log(m)$ \\
\hline
&&&&&&\\

100         &$(1,2) $&$  94.93[0.40](0)  $&$ 95.17[0.55] (0) $&$ 94.68[0.70](0.43)  $&$ (100)$ \\
            &$(50,51) $&$  95.15[0.40](0)  $&$ 95.02[0.46](0)  $&$ 94.78[0.52](0.43) $&$ (100)$ \\
            &$(99,100)$&$  95.66[0.40](0) $&$ 95.17[0.42] (0)$&$ 94.75[0.44] (0.43) $&$ (100)$ \\
&&&&&&\\
\hline
\end{tabular}
\end{table}
\begin{table}[H]\centering
\caption{Estimated coverage probabilities of $\beta_i-\beta_j$ for pair $(i,j)$ as well as the length of confidence intervals (in square brackets), and the probabilities that the MLE does not exist (in parentheses), multiplied by $100$.}
\label{Table 2:Estimatation of beta}
\small
\vskip5pt
\begin{tabular}{ccccccc}
\\
\hline
n      &  $(i,j)$ & $L=0$ & $L=\log n$ & $L=(\log n)^{1/2}$ & $L=\log(n)$ \\
\hline
&&&&&&\\

200         &$(1,2) $&$  94.77[0.60](0)  $&$ 94.13[0.83] (0) $&$ 94.75[1.17](2.41)  $&$ (100)$ \\
            &$(100,101) $&$  95.41[0.60](0)  $&$ 94.79[0.68](0)  $&$ 94.87[0.77](2.41) $&$ (100)$ \\
            &$(198,199)$&$  96.60[0.60](0) $&$ 95.04[0.60] (0)$&$ 94.86[0.63] (2.41) $&$ (100)$ \\
&&&&&&\\
\hline
\end{tabular}
\end{table}

\newpage
\begin{table}[H]
\vspace{-4em}
\centering
\small
\centering \caption{The Student Extracurricular network data: the estimated influence parameters $\hat{\boldsymbol{\theta}}$, $95\%$ confidence intervals (in square brackets) and their standard errors (in parentheses)}
\vspace{0.5em}
\label{Table 3:estimated influence parameters}
\begin{tabular}{lcc}
\hline
 Student Organizations  &     Degree & $\hat{\alpha}_i$ \\
\hline
Spanish.Club &     $199$ & $-0.32[-0.49,-0.16](0.08)$ \\

  Pep.Club &     $157$ & $-0.64[-0.82,-0.46](0.09)$ \\

       NHS &     $124$ & $-0.94[-1.14,-0.75](0.10)$ \\

Latin.Club &      $93$ & $-1.29[-1.51,-1.07](0.11)$ \\

Orchestra..Symphonic &      $89$ & $-1.34[-1.57,-1.12](0.11)$ \\

  Key.Club &      $76$ & $-1.52[-1.76,-1.28](0.12)$ \\

Spanish.Club..high. &      $68$ & $-1.65[-1.90,-1.40](0.13)$ \\

Drunk.Driving &      $67$ & $-1.67[-1.92,-1.41](0.13)$ \\

Forensics..National.Forensics.League. &      $66$ & $-1.68[-1.94,-1.43](0.13)$ \\

Choir..a.capella &      $65$ & $-1.70[-1.96,-1.44](0.13)$ \\

         $\vdots$ &            $\vdots$&  $\vdots$          \\

Chess.Club &       $7$ & $-4.04[-4.78,-3.29](0.38)$ \\

Volleyball..JV &       $7$ & $-4.04[-4.78,-3.29](0.38)$ \\

Teachers.of.Tomorrow &       $5$ & $-4.38[-5.26,-3.50](0.45)$ \\

Quiz.Bowl..all. &       $5$ & $-4.38[-5.26,-3.50](0.45)$ \\

Cheerleaders..Spirit.Squad &       $5$ & $-4.38[-5.26,-3.50](0.45)$ \\

Drunk.Driving.Officers &       $4$ & $-4.60[-5.58,-3.62](0.50)$ \\

Choir..vocal.ensemble..4.women. &       $4$ & $-4.60[-5.58,-3.62](0.50)$ \\

Choir..barbershop.quartet..4.men. &       $4$ & $-4.60[-5.58,-3.62](0.50)$ \\

Cross.Country..girls.8th &       $3$ & $-4.89[-6.02,-3.76](0.58)$ \\

Swim...Dive.Team..boys &       $3$ & $-4.89[-6.02,-3.76](0.58)$ \\
\hline
\end{tabular}

\end{table}
\newpage
\begin{table}[H]
\vspace{-4em}
\centering
\small
\centering \caption{The Student Extracurricular network data: the estimated influence parameters $\hat{\boldsymbol{\theta}}$, $95\%$ confidence intervals (in square brackets) and their standard errors (in parentheses)}
\vspace{0.5em}
\label{Table 4:estimated influence parameters}
\begin{tabular}{ccc}
\hline
         Student ID &     Degree & $\hat{\beta}_j$ \\
\hline
 $122662$ &      $14$ & $0.996[0.387,1.606](0.311)$ \\

 $114850$ &      $12$ & $0.792[0.146,1.438](0.330)$ \\

 $888947$ &      $12$ & $0.792[0.146,1.438](0.330)$ \\

 $126259$ &      $11$ & $0.680[0.011,1.348](0.341)$ \\

 $139161$ &      $11$ & $0.680[0.011,1.348](0.341)$ \\

 $122638$ &      $10$ & $0.559[-0.136,1.254](0.355)$ \\

 $888981$ &      $10$ & $0.559[-0.136,1.254](0.355)$ \\

 $888988$ &      $10$ & $0.559[-0.136,1.254](0.355)$ \\

 $889059$ &      $10$ & $0.559[-0.136,1.254](0.355)$ \\

 $114037$ &       $9$ & $0.428[-0.298,1.153](0.370)$ \\

 $114671$ &       $9$ & $0.428[-0.298,1.153](0.370)$ \\

 $888892$ &       $9$ & $0.428[-0.298,1.153](0.370)$ \\

 $\vdots$ &  $\vdots$ &  $\vdots$ \\
 $889102$ &       $1$ & $-1.985[-3.974,0.004](1.015)$ \\

 $889103$ &       $1$ & $0.000[-0.844,0.844](0.431)$ \\
\hline
\multicolumn{3}{l}{The full table is available by sending emails to \url{zhang_yong@mails.ccnu.edu.cn}}
\end{tabular}
\end{table}

\end{document}